\documentclass[
]{ceurart}

\usepackage{listings}
\usepackage{graphbox}
\DeclareRobustCommand{\DLLogo}{%
  \begingroup\normalfont
  \kern-1.75pt\includegraphics[align=c,height=1.25\baselineskip]{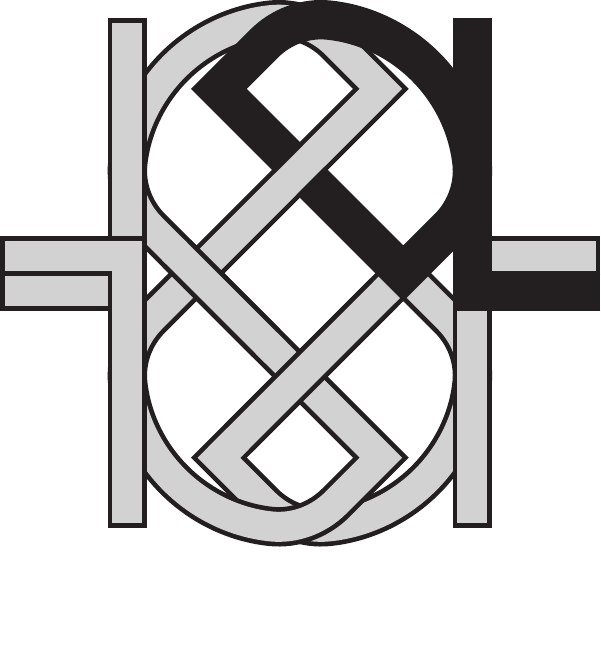}\kern-1.5pt%
  \endgroup
}

\usepackage{amsthm}
\newtheorem{theorem}{Theorem}
\newtheorem{definition}{Definition}
\newtheorem{example}{Example}

  \newtheorem{innercustomthm}{Theorem}
\newenvironment{customthm}[1]
  {\renewcommand\theinnercustomthm{#1}\innercustomthm}
  {\endinnercustomthm}

\newcommand{\ifandonlyif}{\textit{iff} }
\newcommand{\iffi}{\textit{iff} }
\newcommand{\resp}{resp.}
\newcommand{\sect}{Section}
\newcommand{\fig}{Figure}

\newcommand{\thm}{Theorem}

\newcommand{\ih}{IH}

\definecolor{tim}{RGB}{0, 0, 250}

\newcommand{\gtalc}{\mathsf{G3}\alc}
\newcommand{\gtalce}{\mathsf{G3}\alc^{\star}}
\newcommand{\gtalcee}{\mathsf{G3}\alc^{\ast}}

\newcommand{\cnames}{\mathbf{C}}
\newcommand{\rnames}{\mathbf{R}}
\newcommand{\inames}{\mathbf{I}}
\newcommand{\tbox}{\mathcal{T}}
\newcommand{\abox}{\mathcal{A}}

\newcommand{\branch}{\mathcal{B}}

\newcommand{\fcomp}[1]{\|#1\|}
\newcommand{\ec}[1]{[#1]_{\sim}}
\newcommand{\rrel}[1]{\mathsf{Rel}(#1)}

\newcommand{\kb}{\mathcal{K}}
\newcommand{\inv}[1]{{#1}^{-}}

\newcommand{\trans}[1]{\mathsf{Trans}(#1)}
\newcommand{\comp}{\circ}
\newcommand{\eq}{\approx}
\newcommand{\ineq}{\not\eq}
\newcommand{\univ}{\mathsf{U}}
\newcommand{\self}{\mathsf{Self}}
\newcommand{\irr}[1]{\mathsf{Irr}(#1)}
\newcommand{\refl}[1]{\mathsf{Refl}(#1)}
\newcommand{\asy}[1]{\mathsf{Asy}(#1)}
\newcommand{\funct}[1]{\mathsf{Funct}(#1)}
\newcommand{\disj}[1]{\mathsf{Dis}(#1)}
\newcommand{\vocab}{\mathcal{V}}
\newcommand{\card}[1]{\#{#1}}

\newcommand{\alc}{\mathcal{ALC}}
\newcommand{\s}{\mathcal{S}}
\newcommand{\h}{\mathcal{H}}
\newcommand{\sr}{\mathcal{SR}}
\newcommand{\nom}{\mathcal{O}}
\newcommand{\ir}{\mathcal{I}}
\newcommand{\f}{\mathcal{F}}
\newcommand{\n}{\mathcal{N}}
\newcommand{\q}{\mathcal{Q}}

\newcommand{\inter}{\mathcal{I}}
\newcommand{\dom}{\Delta^{\inter}}
\newcommand{\map}[1]{{#1}^{\inter}}

\newcommand{\cxtl}{\Sigma}
\newcommand{\cxtr}{\Pi}
\newcommand{\rcxtl}{\mathpzc
{R}}
\newcommand{\rcxtr}{\mathpzc
{Q}}
\newcommand{\sar}{\vdash}
\newcommand{\lrf}{F}
\newcommand{\lrff}{F}
\newcommand{\lrfg}{G}

\newcommand{\con}{\sqcap}
\newcommand{\dis}{\sqcup}
\newcommand{\imp}{\sqsubseteq}
\newcommand{\all}{\forall r.}
\newcommand{\some}{\exists r.}

\newcommand{\idc}{(id_{\cnames})}
\newcommand{\idr}{(id_{\rnames})}
\newcommand{\botl}{(\bot_{l})}
\newcommand{\topr}{(\top_{r})}
\newcommand{\botr}{(\bot_{r})}
\newcommand{\topl}{(\top_{l})}
\newcommand{\negl}{(\neg_{l})}
\newcommand{\negr}{(\neg_{r})}
\newcommand{\conl}{(\con_{l})}
\newcommand{\conr}{(\con_{r})}
\newcommand{\disl}{(\dis_{l})}
\newcommand{\disr}{(\dis_{r})}
\newcommand{\impl}{(\imp_{l})}
\newcommand{\impr}{(\imp_{r})}
\newcommand{\alll}{(\forall_{l})}
\newcommand{\allr}{(\forall_{r})}
\newcommand{\existsl}{(\exists_{l})}
\newcommand{\existsr}{(\exists_{r})}

\newcommand{\invi}{(inv(r)_{l})}
\newcommand{\invii}{(inv(\inv{r})_{l})}
\newcommand{\inviii}{(inv(r)_{r})}
\newcommand{\inviv}{(inv(\inv{r})_{r})}

\newcommand{\ltnl}{(\leqslant n r . P_{l})}
\newcommand{\ltnr}{(\leqslant n r . P_{r})}
\newcommand{\gtnl}{(\geqslant n r . P_{l})}
\newcommand{\gtnr}{(\geqslant n r . P_{r})}

\newcommand{\ltunl}{(\leqslant n r_{l})}
\newcommand{\ltunr}{(\leqslant n r_{r})}
\newcommand{\gtunl}{(\geqslant n r_{l})}
\newcommand{\gtunr}{(\geqslant n r_{r})}

\newcommand{\compl}{(\comp_{l})}
\newcommand{\compr}{(\comp_{r})}
\newcommand{\transl}{(\mathsf{Trans}(r)_{l})}
\newcommand{\transr}{(\mathsf{Trans}(r)_{r})}

\newcommand{\crial}{(cria_{l})}
\newcommand{\criar}{(cria_{r})}
\newcommand{\refll}{(\refl{r}_{l})}
\newcommand{\reflr}{(\refl{r}_{r})}
\newcommand{\irrl}{(\irr{r}_{l})}
\newcommand{\irrr}{(\irr{r}_{r})}
\newcommand{\asyl}{(\asy{r}_{l})}
\newcommand{\asyr}{(\asy{r}_{r})}
\newcommand{\disjl}{(\disj{r,s}_{l})}
\newcommand{\disjr}{(\disj{r,s}_{r})}
\newcommand{\noml}{(\{b\}_{l}^{1})}
\newcommand{\nomli}{(\{b\}_{l}^{1})}
\newcommand{\nomr}{(\{b\}_{r}^{1})}
\newcommand{\nomrii}{(\{b\}_{r}^{2})}
\newcommand{\nomlii}{(\{b\}_{l}^{2})}

\newcommand{\ineql}{(\ineq_{l})}
\newcommand{\ineqr}{(\ineq_{r})}
\newcommand{\eql}{(\eq_{l})}
\newcommand{\eqr}{(\eq_{r})}

\newcommand{\eqeucl}{(\mathsf{Euc}(\eq))}

\newcommand{\rnegl}{(\neg \rnames_{l})}
\newcommand{\rnegr}{(\neg \rnames_{r})}
\newcommand{\univl}{(\univ_{l})}
\newcommand{\univr}{(\univ_{r})}
\newcommand{\selfl}{(\self_{l})}
\newcommand{\selfr}{(\self_{r})}

\newcommand{\functl}{(\funct{r}_{l})}

\newcommand{\eqrepli}{(\mathsf{Rep}_{1}(\eq))}
\newcommand{\eqreplii}{(\mathsf{Rep}_{2}(\eq))}

\newcommand{\ddrl}{(\mathsf{Rel}_{l})}
\newcommand{\ddrr}{(\mathsf{Rel}_{r})}

\newcommand{\sub}{(sub)}

\newcommand{\wkl}{(wk_{l})}
\newcommand{\wkr}{(wk_{r})}

\newcommand{\ctrl}{(ctr_{l})}
\newcommand{\ctrr}{(ctr_{r})}

\usepackage{amssymb, amsmath, mathrsfs}
\usepackage{amsthm}
\usepackage{synttree, bussproofs,scrextend, stmaryrd, rotating, xcolor, algorithm2e, upgreek}
\usepackage{esvect,pgf, tikz, color}
\usetikzlibrary{arrows, automata}
\usepackage[all]{xy}
\usepackage{changepage,enumerate,proof,upgreek,relsize,trfsigns,
comment}
\DeclareMathAlphabet{\mathpzc}{OT1}{pzc}{m}{it}

\begin{document}

\copyrightyear{2022}
\copyrightclause{Copyright for this paper by its authors.
  Use permitted under Creative Commons License Attribution 4.0
  International (CC BY 4.0).}

\conference{\DLLogo{} DL 2022: 35th International Workshop on Description Logics,
  August 7--10, 2022, Haifa, Israel}

\title{Uniform and Modular Sequent Systems for Description Logics}

\author[1]{Tim Lyon}[%
orcid=0000-0003-3214-0828,
email=timothy_stephen.lyon@tu-dresden.de,
url=https://iccl.inf.tu-dresden.de/web/Tim_Lyon,
]
\address[1]{Computational Logic Group, Institute of Artificial Intelligence, Technische Universität Dresden,
  01062 Dresden, Germany}

\author[1]{Jonas Karge}[%
orcid=0000-0002-8296-1010,
email=jonas.karge@tu-dresden.de,
url=https://iccl.inf.tu-dresden.de/web/Jonas_Karge,
]

\begin{abstract}
We introduce a framework that allows for the construction of sequent systems for expressive description logics extending $\alc$. Our framework not only covers a wide array of common description logics, but also allows for sequent systems to be obtained for extensions of description logics with special formulae that we call \emph{role relational axioms}. All sequent systems are sound, complete, and possess favorable properties such as height-preserving admissibility of common structural rules and height-preserving invertibility of rules.
\end{abstract}


\maketitle

\section{Introduction}



Description logics (DLs) consist of an assortment of knowledge representation languages used to structure and represent knowledge in an unequivocal and perspicuous manner. 
 In DLs, knowledge is represented by means of \emph{knowledge bases (KBs)}, i.e. collections of expressions involving concepts and roles. KBs contain explicit knowledge of a particular domain of interest, and by means of logical consequence, implicit knowledge may be derived, thus giving rise to a need for logical tools to extract information. 
 In addition, it is reasonable to request that such tools be \emph{automatable}, i.e. it is not only desirable to develop tools that have the potential of deriving information, but which give definitive answers to a problem by means of an algorithm.  It is also worthwhile to possess tools that allow one to constructively prove (meta-)logical properties of DLs (e.g. concept interpolation, or re-writings of concepts and TBoxes), and 
 which are applicable to a wide array of DLs, regardless of their idiosyncrasies.

Such tools---meeting the above demands---are capable of being developed on the basis of proof theory. Indeed, various DLs have been equipped with tableau-based proof-search algorithms~\cite{BaaHorLutSat17,DonLenNarNut97,HorSat04,OrtCalEit06,OrtCalEit08,SchSmo91}, resolution-based algorithms~\cite{KazMot06,MotSat06,TriStoChoSta15}, or consequence-based algorithms~\cite{SimKazHor11,Kaz09}, to solve certain reasoning tasks. These works highlight and demonstrate the success of proof-theoretic methods in application to problems of description logics. Therefore, 
a proof-theoretic formalism that yields proof systems for a significant number of DLs \emph{on demand} is desirable. 
Hence, the intent of this paper is to propose a uniform and modular framework for generating proof systems---namely, \emph{sequent systems}---for a large class of DLs, in the style of~\cite{NegPla11}. That is, the purpose of this paper is to provide a general recipe for constructing sequent systems for DLs.

Although work has been done on supplying sequent systems for DLs~\cite{BorFraHor00,Hof05,Rad12, Str97}, the systems have been constructed for a relatively narrow set. The distinguishing feature of the present paper is that we provide a formalism for generating sound and complete sequent systems for a sizable class of expressive DLs. Indeed, our work not only covers $\alc$ and its prominent extensions (e.g. $\s\h\ir\nom\q$ and the DL $\sr\nom\ir\q$ that underlies OWL 2~\cite{GraHorMotParPatSat08}), but allows for extensions of expressive DLs with axioms we refer to as \emph{role relational axioms (RRAs)}. 
Such axioms express properties of, and relationships between, roles. For instance, $\trans{r}$ and $\disj{r,s}$, which express that the role $r$ is transitive and the roles $r$ and $s$ are disjoint, respectively, are defined to be instances of role relational axioms. It will be seen that the sequent formalism we provide is both uniform, covering many DLs, and modular, meaning that a sequent system for one DL is straightforwardly transformable into a sequent system for another DL by the addition or deletion of inference rules. Due to space constraints we leave the discussion of complexity related issues as well as proof-search algorithms up to future work.

The paper is organized as follows: In (\sect~\ref{sec:log-prelims}), we introduce expressive DLs, including their semantics and features of their knowledge bases. In (\sect~\ref{sec:sequent-systems}), we introduce a sequent calculus for the \emph{attributive concept language with complements} $\alc$~\cite{SchSmo91}, and define extensions for other expressive DLs along with the addition of rules for RRAs. We argue that all of our sequent calculi are sound, complete, and possess standard properties (e.g. invertibility of rules and admissibility of contraction).


\label{sec:introduction}

\section{Description Logics}\label{sec:log-prelims}

In this section, we present the family of expressive description logics (DLs) (cf.~\cite{OrtSim12}) that will be considered in this paper. This class of logics is obtained by extending 
 $\alc$. 
We first define $\alc$ and its associated semantics, and then discuss extensions thereof.

\subsection{Preliminaries and $\alc$}


$\alc$, and DLs more generally, are defined relative to a \emph{vocabulary} $\vocab = (\rnames, \cnames, \inames)$ the components of which are taken to be pairwise disjoint, countable sets. Each set contains primitive symbols dedicated to a particular purpose: the set $\rnames$ contains \emph{role names} used to denote binary relations, the set $\cnames$ contains \emph{concept names} used to denote classes of entities, and the set $\inames$ contains \emph{individuals} used to denote particular entities. We use $r$, $s$, $\ldots$ (potentially annotated) to denote role names, $C$, $D$, $\ldots$ (potentially annotated) to denote concept names, and $a$, $b$, $\ldots$ (potentially annotated) to denote individuals. For $\alc$, \emph{complex concepts} are built from role and concept names as dictated by the following BNF grammar:
$$
P ::= C \ | \ \bot \ | \ \top \ | \ \neg P \ | \ P \sqcup P \ | \ P \sqcap P \ | \ \exists r . P \ | \ \forall r . P
$$
where $C \in \cnames$ and $r \in \rnames$. We use the symbols $P$, $Q$, $\ldots$ (potentially annotated) to denote complex concepts. We interpret complex concepts and roles as follows:


\begin{definition}[Interpretation~\cite{BaaHorLutSat17}] An \emph{interpretation} $\inter = (\dom,\map{\cdot})$ contains a non-empty set $\dom$, called the \emph{domain}, and a map $\map{\cdot}$ such that for every $C \in \cnames$, $\map{C} \subseteq \dom$; for every $r \in \rnames$, $\map{r} \subseteq \dom \times \dom$; and for every $a \in \inames$, $\map{a} \in \dom$. The map $\map{\cdot}$ is extended to complex concept names as follows:






\medskip

$\map{\top} := \dom$;  $\map{\bot} := \emptyset$; $\map{C \sqcup D} := \map{C} \cup \map{D}$;  $\map{C \sqcap D} := \map{C} \cap \map{D}$; 

$\map{\exists r. C} := \{a \in \dom \ | \ \text{ there exists $b \in \dom$ s.t. $(a,b) \in \map{r}$ and $b \in \map{C}$.}\}$; 

$\map{\forall r. C} := \{a \in \dom \ | \ \text{ for each $b \in \dom$, if $(a,b) \in \map{r}$, then $b \in \map{C}$.}\}$.









\end{definition}

As is standard for DLs, we collect specific formulae into \emph{TBoxes} to specify certain properties of, and relationships between, concepts and roles. For $\alc$, a TBox 
is a finite set of \emph{general concept inclusions (GCIs)}, which are formulae of the form $P \sqsubseteq Q$, where $P$ and $Q$ are complex concepts. As explained in the following section (\sect~\ref{subsec:extensions-of-ALC}), we allow for a larger variety of formulae in TBoxes for DLs more expressive than $\alc$.

Typically, for DLs, assertional knowledge is represented by formulae that state whether or not an individual or pair of individuals participate in a concept or role. Such formulae, which are referred to as \emph{assertions}, comprise the \emph{ABox}. For $\alc$, the ABox 
contains a finite number of \emph{concept assertions} of the form $a : P$ (with $P$ a complex concept and $a \in \inames$) and a finite number of \emph{role assertions} of the form $r(a,b)$ (with $r \in \rnames$ and $a,b \in \inames$). A \emph{knowledge base (KB)} $\kb$
is defined to be a pair consisting of a TBox $\tbox$ and an ABox $\abox$, i.e. $\kb = (\tbox,\abox)$. Let us now define how interpretations can be extended to the formulae of TBoxes, ABoxes, and therefore, to KBs. 

\begin{definition}[Model~\cite{BaaHorLutSat17}] An interpretation $\inter = (\dom,\map{\cdot})$ \emph{satisfies} a GCI $P \sqsubseteq Q$, written $\inter \models P \sqsubseteq Q$, \ifandonlyif $\map{P} \subseteq \map{Q}$; a concept assertion $a : P$, written $\inter \models a : P$, \ifandonlyif $\map{a} \in \map{P}$; and a role assertion $r(a,b)$, written $\inter \models r(a,b)$, \ifandonlyif $(\map{a},\map{b}) \in \map{r}$. We say that an intepretation $\inter$ is a \emph{model} of a TBox $\tbox$ (ABox $\abox$) \ifandonlyif it satisfies all formulae in $\tbox$ (all formulae in $\abox$, \resp). An interpretation $\inter$ is a \emph{model} of a KB $\kb = (\tbox, \abox)$ \ifandonlyif it is a model of $\tbox$ and $\abox$.
\end{definition}






\subsection{Extensions of $\alc$}\label{subsec:extensions-of-ALC}

The sequent systems provided in the subsequent section allow for a sizable number of DLs to be captured proof-theoretically. We focus our attention on presenting well-known extensions of $\alc$, making use of the well-established naming convention for DLs to do so. Also, we define how new formulae within extensions are satisfied by a given interpretation $\inter = (\dom,\map{\cdot})$.



\underline{{$\s$}} Prepending the name of a DL with $\s$ (rather than $\alc$) indicates that a TBox is permitted to include \emph{transitivity axioms} of the form $\trans{r}$, or equivalently, axioms of the form $r \comp r \imp r$, where the \emph{composition operation $\comp$} is interpreted accordingly (with $r,s \in \rnames$): $\map{(r \comp s)} := $
$$ \{(a,b) \in \dom \times \dom \ | \ \text{there exists a $c \in \dom$ s.t. $(a,c) \in \map{r}$ and $(c,b) \in \map{s}$.} \} $$
$\inter$ \emph{satisfies} $\trans{r}$, written $\inter \models \trans{r}$, \ifandonlyif $\map{r}$ is transitive.\footnote{$\map{r}$ is transitive \iffi for all $a,b,c \in \dom$, if $(a, b), (b, c) \in \map{r}$, then $(a, c) \in \map{r}$.}

\medskip

\underline{$\h$} Including an $\h$ in the name of a DL (e.g. $\alc\h$) indicates that \emph{simple role inclusions axioms (RIAs)} of the form $r \imp s$ with $r,s \in \rnames$ may be included in a TBox.
$\inter$ \emph{satisfies} $r \imp s$, written $\inter \models r \imp s$, \ifandonlyif $\map{r} \subseteq \map{s}$.

\medskip

\underline{$\sr$} The most notable feature of DLs whose names are prepended with $\sr$ is that such logics allow for \emph{complex role inclusion axioms (CRIAs)} of the form $r_{1} \comp \cdots \comp r_{n} \imp r$ to be included in a TBox.\footnote{We note that syntactic conditions are usually imposed on the form of CRIAs in order to ensure the decidability of the resulting DL (e.g., see~\cite{BaaHorLutSat17,HorSat04}).} 
 Additionally, DLs from the $\sr$ family may include \emph{reflexivity axioms} of the form $\refl{r}$, \emph{irreflexivity axioms} of the form $\irr{r}$, \emph{asymmetry axioms} of the form $\asy{r}$, or \emph{disjointness axioms} of the form $\disj{r,s}$.\footnote{Each property is defined as follows: (i) $\map{r}$ is reflexive \iffi for each $a \in \dom$, $(a, a) \in \map{r}$, (ii) $\map{r}$ is irreflexive \iffi for each $a \in \dom$, $(a, a) \not\in \map{r}$, (iii) $\map{r}$ is asymmetric \iffi for each $a,b \in \dom$, if $(a, b) \in \map{r}$, then $(b, a) \not\in \map{r}$, and (iv) $\map{r}$ and $\map{s}$ are disjoint \iffi $\map{r} \cap \map{s} = \emptyset$.}

\begin{itemize}
\item $\inter$ \emph{satisfies} $r_{1} \comp \cdots \comp r_{n} \imp r$, written $\inter \models r_{1} \comp \cdots \comp r_{n} \imp r$, \ifandonlyif $\map{r_{1}} \comp \cdots \comp \map{r_{n}} \subseteq \map{r}$;

\item $\inter$ \emph{satisfies} $\refl{r}$, written $\inter \models \refl{r}$, \ifandonlyif $\map{r}$ is reflexive;

\item $\inter$ \emph{satisfies} $\irr{r}$, written $\inter \models \irr{r}$, \ifandonlyif $\map{r}$ is irreflexive;

\item $\inter$ \emph{satisfies} $\asy{r}$, written $\inter \models \asy{r}$, \ifandonlyif $\map{r}$ is asymmetric;

\item $\inter$ \emph{satisfies} $\disj{r,s}$, written $\inter \models \disj{r,s}$, \ifandonlyif $\map{r}$ and $\map{s}$ are disjoint.
\end{itemize}

\medskip

\underline{$\nom$} Including an $\nom$ in the name of a DL indicates that the set $\cnames$ of concept names includes \emph{nominals} of the form $\{a\}$, for each $a \in \inames$. We interpret nominals accordingly:
$\map{\{a\}} := \{\map{a}\}$.


\medskip

\underline{$\ir$} Including an $\ir$ in the name of a DL indicates that the set $\rnames$ includes \emph{inverse roles} of the form $\inv{r}$, for each $r \in \rnames$. We interpret inverse roles accordingly:
$\map{\inv{r}} := \{(b,a)\ | \ (a,b) \in \map{r}\}$.


\medskip

\underline{$\f$} An $\f$ in the name of a DL indicates that a TBox may include \emph{functionality axioms} of the form $\funct{r}$ for $r \in \rnames$.
$\inter$ \emph{satisfies} $\funct{r}$, written $\inter \models \funct{r}$, \ifandonlyif $\map{r}$ is functional.\footnote{$\map{r}$ is functional \iffi for all $a, b, c \in \dom$, if $(a,b), (a,c) \in \map{r}$, then $b = c$.}

\medskip

\underline{$\n$} The symbol $\n$ is included in the name of a DL when it includes \emph{unqualified number restrictions} of the form $(\leqslant n r. \top)$ or $(\geqslant n r. \top)$ with $r \in \rnames$ among its concepts. We interpret unqualified number restrictions as follows:\footnote{We use $\card{S}$ for a set $S$ to denote the \emph{cardinality} of the set.} $\map{(\leqslant n r . \top)} := \{a \in \dom \ | \ \card{\{b \ | \ (a,b) \in \map{r} \}} \leq n \}$ and $\map{(\geqslant n r . \top)} := \{a \in \dom \ | \ \card{\{b \ | \ (a,b) \in \map{r} \}} \geq n \}$.




\medskip

\underline{$\q$} We use $\q$ to indicate that a DL includes \emph{qualified number restrictions} of the form $(\leqslant n r. P)$ or $(\geqslant n r. P)$ with $r \in \rnames$ among its concepts. We interpret qualified number restrictions accordingly: $\map{(\leqslant n r . P)} := \{a \in \dom \ | \ \card{\{b \ | \ (a,b) \in \map{r} \text{ and } b : P \}} \leq n \}$ and $\map{(\geqslant n r . P)} := \{a \in \dom \ | \ \card{\{b \ | \ (a,b) \in \map{r} \text{ and } b : P \}} \geq n \}$.





\medskip

\underline{Other Extensions} We may also extend $\alc$ by permitting the inclusion of \emph{equality} or \emph{inequality} axioms of the form $a \eq b$ and $a \ineq b$ (\resp) in a TBox, by permitting \emph{negated role assertions} of the form $\neg r(a,b)$ in an ABox, by allowing for the \emph{universal role} $\univ$ to be included in $\rnames$ (interpreted $\map{\univ} := \dom \times \dom$), or by allowing the complex concept $\some \self$ for $r \in \rnames$ (interpreted $\map{(\some \self)} := \{a \ | \ (a,a) \in \map{r} \}$). The semantics of (in)equalities and negated role assertions 
 is as follows:
\begin{itemize}

\item $\inter$ \emph{satisfies} $a \eq b$, written $\inter \models a \eq b$, \ifandonlyif $\map{a} = \map{b}$;

\item $\inter$ \emph{satisfies} $a \ineq b$, written $\inter \models a \ineq b$, \ifandonlyif $\map{a} \neq \map{b}$;

\item $\inter$ \emph{satisfies} $\neg r(a,b)$, written $\inter \models \neg r(a,b)$, \ifandonlyif $(a,b) \not\in \map{r}$.



\end{itemize}



\section{Sequent Systems}\label{sec:sequent-systems}

Our proof systems consist of inference rules that manipulate sequents of the form $\Lambda := \rcxtl, \cxtl \sar \cxtr, \rcxtr$, where $\rcxtl, \cxtl$ is referred to as the \emph{antecedent} and $\cxtr, \rcxtr$ is referred to as the \emph{consequent}. Note that $\cxtl$, $\cxtr$, $\rcxtl$, and $\rcxtr$ are taken to be (potentially empty) multisets of DL formulae. $\cxtl$ and $\cxtr$ are multisets of formulae of the form $a : P$, called \emph{internal formulae (IFs)}, where $a$ ranges over the set of individuals $\inames$, and $P$ is a complex concept generated via the following grammar in BNF:
$$
P ::=  C \ | \  \bot \ | \  \top \ | \  \neg P \ | \  P \dis P \ | \  P \con P \ | \  \some P \ | \  \all P \ | \  \{a\} \ | \  (\leqslant n r. P) \ | \  (\geqslant n r. P) \ | \ \some \self
$$
with $C \in \cnames$, $r \in \rnames$ (which is potentially an inverse role $\inv{s}$ or the universal role $\univ$), $a \in \inames$, and $n \in\mathbb{N}$. $\rcxtl$ and $\rcxtr$ consist of formulae generated via the following grammar in BNF, and are referred to as \emph{external formulae (EFs)}.
$$
\lrf ::= P \imp Q \ | \ r(a,b) \ | \ \neg r(a,b) \ | \ \rrel{r_{1}, \ldots, r_{n}} \ | \ r_{1} \comp \cdots \comp r_{n} \imp r \ | \ a \eq b \ | \ a \ineq b
$$
where $P$ and $Q$ are complex concepts, $a,b \in \inames$, $r_{1}, \ldots, r_{n}, r \in \rnames$ (and are potentially inverse roles $\inv{s}$ or the universal role $\univ$), and for each arity $n \in \mathbb{N}$, the relation name $\mathsf{Rel}$ ranges over a countable set of $n$-ary relation names. We note that transitivity axioms $\trans{r}$, reflexivity axioms $\refl{r}$, irreflexivity axioms $\irr{r}$, asymmetry axioms $\asy{r}$, disjointness axioms $\disj{r,s}$, and functionality axioms $\funct{r}$ are all instances of formulae of the form $\rrel{r_{1}, \ldots, r_{n}}$, which we refer to as \emph{role relational axioms (RRAs)}. We use $\lrff$, $\lrfg$, $\ldots$ to denote EFs defined by the grammar above. We distinguish EFs from IFs as EFs are those formulae which \emph{govern reasoning with complex concepts}, i.e. of reasoning with IFs. 

When supplying a calculus for a particular DL, we assume that the EFs and IFs occurring within sequents are restricted to those formulae allowed by the DL language under consideration. For example, for $\alc$, we omit the inclusion of nominals, (un)qualified number restrictions, and $\some \self$ from occurring in IFs since such concepts are not included in $\alc$'s language.

\subsection{The System $\gtalc$ and Descriptive Definitional Rules}

We now present our calculus $\gtalc$ for the DL $\alc$ as well as define extensions of the calculus with \emph{descriptive definitional rules (DDRs)}.\footnote{For a discussion of $\mathsf{G3}$-style calculi, along with the $\mathsf{G1}$ and $\mathsf{G2}$ variants, see~\cite[Section~80]{Kle52}.} DDRs introduce RRAs into either the antecedent or consequent of a sequent, and thus provide our calculus with the capacity to handle such formulae. We discuss DDRs in detail below, and mention the DDRs that introduce widely-used RRAs such as transitivity axioms and reflexivity axioms. 
The calculus $\gtalc$ is obtained by transforming the semantics of $\alc$ into inference rules (cf.~\cite{NegPla11,Sim94,Vig00}), and is displayed in \fig~\ref{fig:G3ALC}. Note that in the $\idr$ rule we stipulate that $F$ must be of the form $r(a,b)$ or $a \eq b$.
 We refer to the \emph{principal formulae} of a rule as those formulae which are explicitly presented in the conclusion (e.g. $a : P \dis Q$ is the principal formula of $\disl$), and to the multisets $\rcxtl$, $\cxtl$, $\cxtr$, and $\rcxtr$ as \emph{contexts}. Furthermore, we note that proofs/derivations are constructed by successively applying inference rules to \emph{initial rules/sequents}, i.e. rules without premises (e.g. $\idc$, $\idr$, $\botl$, and $\topr$), and the \emph{height} of a proof is defined to be the longest sequence of sequents from the conclusion of the proof to an initial rule (cf.~\cite{NegPla11}).
 

\begin{figure}

\begin{center}
\begin{tabular}{c c}
\AxiomC{}
\RightLabel{$\idc$}
\UnaryInfC{$\rcxtl, \cxtl, a : C \sar a : C, \cxtr, \rcxtr$}
\DisplayProof

&

\AxiomC{}
\RightLabel{$\idr$}
\UnaryInfC{$\rcxtl, \cxtl, \lrf \sar \lrf, \cxtr, \rcxtr$}
\DisplayProof
\end{tabular}
\end{center}

\begin{center}
\begin{tabular}{c c c}
\AxiomC{}
\RightLabel{$\botl$}
\UnaryInfC{$\rcxtl, \cxtl, a : \bot \sar \cxtr, \rcxtr$}
\DisplayProof

&

\AxiomC{$\rcxtl, \cxtl \sar a : \bot, \cxtr, \rcxtr$}
\RightLabel{$\botr$}
\UnaryInfC{$\rcxtl, \cxtl \sar \cxtr, \rcxtr$}
\DisplayProof

&

\AxiomC{$\rcxtl, \cxtl, a : \top \sar \cxtr, \rcxtr$}
\RightLabel{$\topl$}
\UnaryInfC{$\rcxtl, \cxtl \sar \cxtr, \rcxtr$}
\DisplayProof
\end{tabular}
\end{center}

\begin{center}
\begin{tabular}{c c c}

\AxiomC{}
\RightLabel{$\topr$}
\UnaryInfC{$\rcxtl, \cxtl \sar a : \top, \cxtr, \rcxtr$}
\DisplayProof

&

\AxiomC{$\rcxtl, \cxtl \sar a : P, \cxtr, \rcxtr$}
\RightLabel{$\negl$}
\UnaryInfC{$\rcxtl, \cxtl, a : \neg P \sar \cxtr, \rcxtr$}
\DisplayProof

&

\AxiomC{$\rcxtl, \cxtl, a : P \sar \cxtr, \rcxtr$}
\RightLabel{$\negr$}
\UnaryInfC{$\rcxtl, \cxtl \sar a : \neg P, \cxtr, \rcxtr$}
\DisplayProof
\end{tabular}
\end{center}

\begin{center}
\begin{tabular}{c c}
\AxiomC{$\rcxtl, \cxtl, a : P \sar \cxtr, \rcxtr$}
\AxiomC{$\rcxtl, \cxtl, a : Q \sar \cxtr, \rcxtr$}
\RightLabel{$\disl$}
\BinaryInfC{$\rcxtl, \cxtl, a : P \dis Q \sar \cxtr, \rcxtr$}
\DisplayProof

&

\AxiomC{$\rcxtl, \cxtl \sar a : P, a : Q, \cxtr, \rcxtr$}
\RightLabel{$\disr$}
\UnaryInfC{$\rcxtl, \cxtl \sar a : P \dis Q, \cxtr, \rcxtr$}
\DisplayProof
\end{tabular}
\end{center}

\begin{center}
\begin{tabular}{c c}
\AxiomC{$\rcxtl, \cxtl, a : P, a : Q \sar \cxtr, \rcxtr$}
\RightLabel{$\conl$}
\UnaryInfC{$\rcxtl, \cxtl, a : P \con Q \sar \cxtr, \rcxtr$}
\DisplayProof

&

\AxiomC{$\rcxtl, \cxtl \sar a : P, \cxtr, \rcxtr$}
\AxiomC{$\rcxtl, \cxtl \sar a : Q, \cxtr, \rcxtr$}
\RightLabel{$\conr$}
\BinaryInfC{$\rcxtl, \cxtl \sar a : P \con Q, \cxtr, \rcxtr$}
\DisplayProof
\end{tabular}
\end{center}

\begin{center}
\begin{tabular}{c c}
\AxiomC{$\rcxtl, P \imp Q, a : P, a : Q, \cxtl \sar \cxtr, \rcxtr$}
\RightLabel{$\impl$}
\UnaryInfC{$\rcxtl, P \imp Q, a : P, \cxtl \sar \cxtr, \rcxtr$}
\DisplayProof

&

\AxiomC{$\rcxtl, \cxtl, b : P \sar b : Q, \cxtr, \rcxtr$}
\RightLabel{$\impr^{\dag}$}
\UnaryInfC{$\rcxtl, \cxtl \sar P \imp Q, \cxtr, \rcxtr$}
\DisplayProof
\end{tabular}
\end{center}

\begin{center}
\begin{tabular}{c c}
\AxiomC{$\rcxtl, \cxtl, r(a,b), b : P \sar \cxtr, \rcxtr$}
\RightLabel{$\existsl^{\dag}$}
\UnaryInfC{$\rcxtl, \cxtl, a : \some P \sar \cxtr, \rcxtr$}
\DisplayProof

&

\AxiomC{$\rcxtl, \cxtl, r(a,b) \sar a : \some P, b : P, \cxtr, \rcxtr$}
\RightLabel{$\existsr$}
\UnaryInfC{$\rcxtl, \cxtl, r(a,b) \sar a : \some P, \cxtr, \rcxtr$}
\DisplayProof
\end{tabular}
\end{center}

\begin{center}
\begin{tabular}{c c}
\AxiomC{$\rcxtl, \cxtl, r(a,b), a : \all P, b : P \sar \cxtr, \rcxtr$}
\RightLabel{$\alll$}
\UnaryInfC{$\rcxtl, \cxtl, r(a,b), a : \all P \sar \cxtr, \rcxtr$}
\DisplayProof

&

\AxiomC{$\rcxtl, \cxtl, r(a,b) \sar b : P, \cxtr, \rcxtr$}
\RightLabel{$\allr^{\dag}$}
\UnaryInfC{$\rcxtl, \cxtl \sar a : \all P, \cxtr, \rcxtr$}
\DisplayProof
\end{tabular}
\end{center}

\caption{$\gtalc$. $\dag$ stipulates that the rule can be applied only if $b$ is an eigenvariable, i.e. $b$ does not occur in the conclusion of the rule.}
\label{fig:G3ALC}
\end{figure}

DDRs are rules which are equivalent to, and obtained from, \emph{descriptive definitions}. Descriptive definitions define properties of, and relationships between, roles; i.e. they define the necessary and sufficient conditions for which an RRA obtains. For instance, the formula $ \trans{r} \leftrightarrow \forall a \forall b \forall c (r(a,b) \land r(b,c) \rightarrow r(a,c))$ defines the RRA $\trans{r}$ for the role $r$. 


\begin{definition}[Descriptive Definition] A \emph{descriptive definition} is a formula of the form:
$$
\rrel{r_{1}, \ldots, r_{l}} \leftrightarrow \forall \vec{a} (F_{1} \land \cdots \land F_{n} \rightarrow G_{1} \lor \cdots \lor G_{k})
$$
such that each $F_{i}$ and $G_{j}$ is an EF of the form $r(a,b)$ or $a \eq b$, the individuals $\vec{a} := a_{1}, \ldots, a_{m}$ occur within $F_{1} \land \cdots \land F_{n}$ (which is $\top$ if the conjunction is empty) and $G_{1} \lor \cdots \lor G_{k}$ (which is $\bot$ if the disjunction is empty), and where the definiens (to the right of the bi-conditional) only makes reference to the roles $r_{1}$, $\ldots$, $r_{l}$ and/or equalities of the form $a \eq b$ (for $a$ and $b$ in $\vec{a}$). 
\end{definition}

Each descriptive definition of the above form can be transformed into a pair of left and right introduction rules (introducing the RRA $\rrel{r_{1}, \ldots, r_{l}}$) as shown below:
\begin{center}
\AxiomC{$\Big\{ \rcxtl, \rrel{r_{1}, \ldots, r_{l}}, \overline{F}, G_{j}, \cxtl \sar \cxtr, \rcxtr \ | \ 1 \leq j \leq k  \Big\}$}
\RightLabel{$\ddrl$}
\UnaryInfC{$\rcxtl, \rrel{r_{1}, \ldots, r_{l}}, \overline{F}, \cxtl \sar \cxtr, \rcxtr$}
\DisplayProof
\end{center}
\begin{center}
\AxiomC{$\rcxtl, \overline{F}, \cxtl \sar \cxtr, \overline{G}, \rcxtr$}
\RightLabel{$\ddrr^{\dag}$}
\UnaryInfC{$\rcxtl, \cxtl \sar \cxtr, \rrel{r_{1}, \ldots, r_{l}}, \rcxtr$}
\DisplayProof
\end{center}
We let $\overline{F} := F_{1}, \ldots, F_{n}$, $\overline{G} := G_{1},\ldots, G_{k}$ and the side condition $\dag$ states that $\ddrr$ is applicable only if the individuals $\vec{a}$ (the collection of all individuals occurring within $\overline{F}$ and $\overline{G}$) are eigenvariables. (NB. Eigenvariables are individuals that do not occur in the conclusion of a rule, i.e. they are fresh in the premise(s), which ensures the soundness of rule applications; for a discussion on eigenvariables, see~\cite{NegPla11}.) 
 We let $\gtalce$ denote $\gtalc$ extended with any finite number of DDR pairs $\{\ddrl,\ddrr\}$, and note that such extensions give calculi for extensions of $\alc$. For example, if we aim to provide a calculus for the DL $\s$, then our calculus must be capable of reasoning with
transitivity axioms 
i.e. formulae of the form $\trans{r}$ with $r \in \rnames$. 
$\trans{r}$ can be defined by means of a descriptive definition, implying that we can obtain a calculus for the DL $\s$ by extending $\gtalc$ with the two rules shown below. (NB. The side condition $\dag$ states that $a$, $b$, and $c$ must be eigenvariables.)
 
\begin{center}
\AxiomC{$\rcxtl, \trans{r}, r(a,b), r(b,c), r(a,c), \cxtl \sar \cxtr, \rcxtr$}
\RightLabel{$\transl$}
\UnaryInfC{$\rcxtl, \trans{r}, r(a,b), r(b,c) , \cxtl \sar \cxtr, \rcxtr$}
\DisplayProof
\end{center}

\begin{center}
\AxiomC{$\rcxtl, r(a,b), r(b,c), \cxtl \sar \cxtr, r(a,c), \rcxtr$}
\RightLabel{$\transr^{\dag}$}
\UnaryInfC{$\rcxtl, \cxtl \sar \trans{r}, \cxtr, \rcxtr$}
\DisplayProof
\end{center}

Some care must be taken when extending $\gtalc$ with DDRs. It is possible that certain properties of $\gtalc$, such as contraction hp-admissibility (see \thm~\ref{thm:admissible-rules}), are not immediately preserved in extensions of the calculus with DDRs. We apply a solution that is motivated by the work of~\cite{NegPla11}; namely, we can avoid such undesirable circumstances by ensuring that any extension of $\gtalc$ with DDRs adheres to the \emph{closure condition}. (NB. For the remainder of the paper, we assume that every extension of $\gtalc$ satisfies the closure condition.)

\begin{definition}[Closure Condition~\cite{NegPla11}] A calculus with DDRs satisfies the \emph{closure condition} \ifandonlyif for any DDR in the calculus which has a substitution instance containing duplicate principal formulae, the calculus also contains an instance of the rule with the duplicate formulae contracted.\footnote{An example illustrating the closure condition can be found in the appendix.}
\end{definition}

Since only a finite number of substitution instances produce duplicate principal formulae in a DDR, the closure condition will only add a finite number of rules in any extension of $\gtalc$.
 We now define a semantics for sequents as this will be used for soundness and completeness. 

\begin{definition}[Sequent Semantics] Let $\inter = (\dom,\map{\cdot})$ be an interpretation. A sequent $\Lambda := \rcxtl, \cxtl \sar \cxtr,\rcxtr$ is \emph{satisfied} in $\inter$, written $\inter \models \Lambda$, \ifandonlyif if $\inter$ satisfies all formulae in $\rcxtl, \cxtl$, then $\inter$ satisfies some formula in $\rcxtr,\cxtr$. A sequent $\Lambda$ is \emph{falsified} in $\inter$ \ifandonlyif $\inter \not\models \Lambda$, i.e. $\Lambda$ is not satisfied in $\inter$. A sequent $\Lambda$ is \emph{valid}, written $\models \Lambda$, \ifandonlyif it is satisfiable in every interpretation, and is invalid otherwise.
\end{definition}







\subsection{Rules for Extensions of $\alc$}\label{subsect:extensions-G3ALC}

We discuss extensions of $\gtalce$ with rules for deriving new concept assertions (e.g. unqualified number restrictions and nominals) and EFs (e.g. equalities and RIAs). We introduce these additional rules in the same manner as we introduced extensions of $\alc$ in \sect~\ref{subsec:extensions-of-ALC}.

\medskip

\underline{$\s$} If the language of our DL includes role compositions, then the rules $\compl$ and $\compr$ (shown below) should be included in the corresponding  calculus to 
 allow reasoning with role compositions. (NB. $s$ is permitted to be a chain $r_{1} \circ \cdots \circ r_{n}$ of role compositions.)
 Since we can use axioms of the form $r \comp r \imp r$ or $\trans{r}$ to indicate that a role $r$ is transitive, there are two distinct sets of rules which can be included in a calculus to allow reasoning with transitive roles.

First, if our DL allows for axioms of the form $r \comp r \imp r$, then the composition rules, and restricted versions of the $\crial$ and $\criar$ rules (introduced in the $\sr$ subsection below) that only allow principal formulae of the form $r \comp r \imp r$, should be included in the corresponding calculus. (NB. The side condition $\dag$ on the $\compl$ rule stipulates that $b$ is an eigenvariable.)

\begin{center}
\AxiomC{$\rcxtl, r(a,b), s(b,c), \cxtl \sar \cxtr, \rcxtr$}
\RightLabel{$\compl^{\dag}$}
\UnaryInfC{$\rcxtl, (r \comp s)(a,c), \cxtl \sar \cxtr, \rcxtr$}
\DisplayProof
\end{center}

\begin{center}
\AxiomC{$\rcxtl, \cxtl \sar \cxtr, (r \comp s)(a,c), r(a,b), \rcxtr$}
\AxiomC{$\rcxtl, \cxtl \sar \cxtr, (r \comp s)(a,c), s(b,c), \rcxtr$}
\RightLabel{$\compr$}
\BinaryInfC{$\rcxtl, \cxtl \sar \cxtr, (r \comp s)(a,c), \rcxtr$}
\DisplayProof
\end{center}

Second, if we make use of transitivity axioms of the form $\trans{r}$ in our DL, then the DDRs $\transl$ and $\transr$, introduced in the previous section, should be included in our calculus to ensure sound and complete reasoning with such formulae.

\medskip

\underline{$\h$} If we wish to enable reasoning with RIAs of the form $r \imp s$ (e.g. as in $\alc\h$), then one should add restricted versions of the $\crial$ and $\criar$ rules (introduced in the $\sr$ subsection below) where $n = 1$, to ensure sound and complete reasoning with RIAs.

\medskip

\underline{$\sr$} To enable reasoning with CRIAs, the composition rules $\compl$ and $\compr$ should be included along with the following $\crial$ and $\criar$ rules. 
(NB. The side condition $\dag$ on the $\criar$ rule states that $a$ and $b$ must be eigenvariables. For readability, let $F$ denote $r_{1} \comp \cdots \comp r_{n} \imp r$.)

\begin{center}
\AxiomC{$\rcxtl,F, \cxtl \sar \cxtr, (r_{1} \comp \cdots \comp r_{n})(a,b), \rcxtr$}
\AxiomC{$\rcxtl, r(a,b), F, \cxtl \sar \cxtr, \rcxtr$}
\RightLabel{$\crial$}
\BinaryInfC{$\rcxtl, F, \cxtl \sar \cxtr, \rcxtr$}
\DisplayProof
\end{center}

\begin{center}
\AxiomC{$\rcxtl, (r_{1} \comp \cdots \comp r_{n})(a,b), \cxtl \sar  \cxtr, r(a,b), \rcxtr$}
\RightLabel{$\criar^{\dag}$}
\UnaryInfC{$\rcxtl, \cxtl \sar \cxtr,F, \rcxtr$}
\DisplayProof
\end{center}

The (ir)reflexivity, asymmetry, and disjointness axioms can all be defined by means of descriptive definitions: $\refl{r} \leftrightarrow \forall a (\top \rightarrow r(a,a))$, $\asy{r} \leftrightarrow \forall a \forall b (r(a,b) \land r(b,a) \rightarrow \bot)$, $\irr{r} \leftrightarrow \forall a (r(a,a) \rightarrow \bot)$, and $\disj{r,s} \leftrightarrow \forall a \forall b (r(a,b) \land s(a,b) \rightarrow \bot)$. Thus, extending $\gtalce$ with the corresponding DDRs provides our calculus with the capacity to reason with such axioms. All such DDRs can be obtained from the $\ddrl$ and $\ddrr$ 
 rule schemata.

\medskip


\underline{$\nom$} To enable reasoning with nominals, one should include the following rules along with the equality rules of the final subsection below.

\begin{center}
\begin{tabular}{c c}
\AxiomC{$\rcxtl, a \eq b, a : \{b\}, \cxtl \sar \cxtr, \rcxtr$}
\RightLabel{$\noml$}
\UnaryInfC{$\rcxtl, a : \{b\}, \cxtl \sar \cxtr, \rcxtr$}
\DisplayProof

&

\AxiomC{$\rcxtl, \cxtl \sar \cxtr, a : \{b\}, a \eq b, \rcxtr$}
\RightLabel{$\nomr$}
\UnaryInfC{$\rcxtl, \cxtl \sar \cxtr, a : \{b\}, \rcxtr$}
\DisplayProof
\end{tabular}
\end{center}

\begin{center}
\begin{tabular}{c c}
\AxiomC{$\rcxtl, b : \{b\}, \cxtl \sar \cxtr, \rcxtr$}
\RightLabel{$\nomlii$}
\UnaryInfC{$\rcxtl, \cxtl \sar \cxtr, \rcxtr$}
\DisplayProof

&

\AxiomC{ }
\RightLabel{$\nomrii$}
\UnaryInfC{$\rcxtl, \cxtl \sar \cxtr, b : \{b\}, \rcxtr$}
\DisplayProof
\end{tabular}
\end{center}

\medskip


\underline{$\ir$} To add support for reasoning with inverse roles, one should not only allow inverse roles to appear in the relevant rules of the calculus (e.g. $\idr$, $\existsl$, and $\allr$), but should also include the following two rules that encode the fact that the roles $r$ and $\inv{r}$ are inverses.

\begin{center}
{
\begin{tabular}{c c}
\AxiomC{$\rcxtl, r(a,b), \inv{r}(b,a), \cxtl \sar \cxtr, \rcxtr$}
\RightLabel{$\invi$}
\UnaryInfC{$\rcxtl, r(a,b), \cxtl\sar \cxtr, \rcxtr$}
\DisplayProof

&

\AxiomC{$\rcxtl, \inv{r}(a,b), r(b,a), \cxtl \sar \cxtr, \rcxtr$}
\RightLabel{$\invii$}
\UnaryInfC{$\rcxtl, \inv{r}(a,b), \cxtl \sar \cxtr, \rcxtr$}
\DisplayProof
\end{tabular}
}
\end{center}

\begin{center}
{
\begin{tabular}{c c}
\AxiomC{$\rcxtl, \cxtl \sar \cxtr, r(a,b), \inv{r}(b,a), \rcxtr$}
\RightLabel{$\inviii$}
\UnaryInfC{$\rcxtl, \cxtl \sar \cxtr, r(a,b), \rcxtr$}
\DisplayProof

&

\AxiomC{$\rcxtl, \cxtl \sar \cxtr, \inv{r}(a,b), r(b,a), \rcxtr$}
\RightLabel{$\inviv$}
\UnaryInfC{$\rcxtl, \cxtl \sar \cxtr, \inv{r}(a,b), \rcxtr$}
\DisplayProof
\end{tabular}
}
\end{center}

\medskip


\underline{$\f$} Functionality axioms of the form $\funct{r}$ can be defined by means of descriptive definitions; e.g. 
$\funct{r} \leftrightarrow \forall a \forall b \forall c (r(a,b) \land r(a,c) \rightarrow b \eq c)$. 
 We can make use of the $\ddrl$ and $\ddrr$ rule schemata to define DDRs for $\funct{r}$. Hence, a calculus can be enabled to reason about functionality axioms by including the equality rules (introduced in final subsection below) along with the pair of DDRs obtained from the above descriptive defintion. 


\medskip

\underline{$\n$} To allow reasoning with unqualified number restrictions, one makes use of versions of the $\ltnl$, $\ltnr$, $\gtnl$, and $\gtnr$ rules (shown in the next subsection $\q$) where the first set of premises is omitted, and where the $b_{i} : P$ formulae are omitted from the remaining premises. We refer to each of these versions as $\ltunl$, $\ltunr$, $\gtunl$, and $\gtunr$, respectively. 
 Additionally, the equality rules of the final subsection below should be included to ensure proper reasoning with equalities.

\medskip

\underline{$\q$} To enable a calculus to derive theorems concerning qualified number restrictions, we add the following four rules along with the equality rules of the final subsection below. (NB. In the $\ltnr$ rule, $\dag_{1}$ states that $b_{0}, \ldots, b_{n}$ must be eigenvariables and $\rcxtr' := \{b_{i} \eq b_{j} \ | \ 0 \leq i < j \leq n \}$, and in the $\gtnl$ rule, $\dag_{2}$ states that $b_{1}, \ldots, b_{n}$ must be eigenvariables and $\rcxtr' := \{b_{i} \eq b_{j} \ | \ 1 \leq i < j \leq n \}$.)

\begin{center}
{
\AxiomC{$\Big \{\rcxtl, r(a,b_{0}), \ldots, r(a,b_{n}), \cxtl, a : (\leqslant n r . P) \sar b_{i} : P, \cxtr, \rcxtr \ | \ 0 \leq i \leq n \Big \} \cup$}
\noLine
\UnaryInfC{$\Big \{\rcxtl, b_{i} \eq b_{j}, r(a,b_{0}), \ldots, r(a,b_{n}), \cxtl, a : (\leqslant n r . P) \sar \cxtr, \rcxtr \ | \ 0 \leq i < j \leq n \Big \}$}
\RightLabel{$\ltnl$}
\UnaryInfC{$\rcxtl, r(a,b_{0}), \ldots, r(a,b_{n}), \cxtl, a : (\leqslant n r . P) \sar \cxtr, \rcxtr$}
\DisplayProof
}
\end{center}

\begin{center}
\AxiomC{$\rcxtl, r(a,b_{0}), \ldots, r(a,b_{n}), \cxtl, b_{0} : P, \ldots, b_{n} : P \sar \cxtr, \rcxtr', \rcxtr$}
\RightLabel{$\ltnr^{\dag_{1}}$}
\UnaryInfC{$\rcxtl, \cxtl \sar a : (\leqslant n r . P), \cxtr, \rcxtr$}
\DisplayProof
\end{center}

\begin{center}
\AxiomC{$\rcxtl, r(a,b_{1}), \ldots, r(a,b_{n}), \cxtl, b_{1} : P, \ldots, b_{n} : P \sar \cxtr, \rcxtr', \rcxtr$}
\RightLabel{$\gtnl^{\dag_{2}}$}
\UnaryInfC{$\rcxtl, \cxtl, a : (\geqslant n r . P) \sar \cxtr, \rcxtr$}
\DisplayProof
\end{center}

\begin{center}
{
\AxiomC{$\Big\{ \rcxtl, r(a,b_{1}), \ldots, r(a,b_{n}), \cxtl \sar b_{i} : P, a : (\geqslant n r . P), \cxtr, \rcxtr \ |\ 1 \leq i \leq n \Big\} \cup$}
\noLine
\UnaryInfC{$\Big \{\rcxtl, b_{i} \eq b_{j}, r(a,b_{1}), \ldots, r(a,b_{n}), \cxtl, a : (\leqslant n r . P) \sar \cxtr, \rcxtr \ | \ 0 \leq i < j \leq n \Big \}$}
\RightLabel{$\gtnr$}
\UnaryInfC{$\rcxtl, r(a,b_{1}), \ldots, r(a,b_{n}), \cxtl \sar a : (\geqslant n r . P), \cxtr, \rcxtr$}
\DisplayProof
}
\end{center}

\medskip
\underline{Other Extensions} To enable reasoning with equalities, we include $\eql$, $\eqr$, $\eqrepli$, $\eqreplii$ and $\eqeucl$; to enable reasoning with inequalities, we add the $\ineql$ and $\ineqr$ rules along with the previous five. To enable reasoning with negated role assertions we include $\rnegl$ and $\rnegr$ in our calculus; to ensure theorems can be derived concerning the universal role $\univ$, we allow the role to be used in the relevant rules of our calculus (e.g. $\idr$, $\existsl$, and $\allr$) and also include the $\univl$ and $\univr$ rules shown below. Last, we include the $\selfl$ and $\selfr$ rules if we want our calculus to support complex concepts of the form $\some \self$. (NB. In the $\eqrepli$ and $\eqreplii$ rules, $[a/b]$ denotes a substitution of $b$ for $a$ in the relevant formula.)

\begin{center}
\begin{tabular}{c c c}
\AxiomC{$\rcxtl, a \eq a, \cxtl \sar \cxtr, \rcxtr$}
\RightLabel{$\eql$}
\UnaryInfC{$\rcxtl, \cxtl \sar \cxtr, \rcxtr$}
\DisplayProof

&

\AxiomC{ }
\RightLabel{$\eqr$}
\UnaryInfC{$\rcxtl, \cxtl \sar \cxtr, a \eq a, \rcxtr$}
\DisplayProof

&

\AxiomC{$\rcxtl, \cxtl \sar \cxtr, r(a,b), \rcxtr$}
\RightLabel{$\rnegl$}
\UnaryInfC{$\rcxtl, \neg r(a,b), \cxtl \sar \cxtr, \rcxtr$}
\DisplayProof
\end{tabular}
\end{center}

\begin{center}
\begin{tabular}{c c}
\AxiomC{$\rcxtl, a \eq b, \cxtl, a : P, b : P \sar \cxtr, \rcxtr$}
\RightLabel{$\eqrepli$}
\UnaryInfC{$\rcxtl, a \eq b, \cxtl, a : P \sar \cxtr, \rcxtr$}
\DisplayProof

&

\AxiomC{$\rcxtl, \cxtl \sar \cxtr, r(a,a), \rcxtr$}
\RightLabel{$\selfr$}
\UnaryInfC{$\rcxtl, \cxtl \sar a : \some \self, \cxtr, \rcxtr$}
\DisplayProof
\end{tabular}
\end{center}




\begin{center}
\begin{tabular}{c c}
\AxiomC{$\rcxtl, a \eq b, \lrf, \lrf[a/b], \cxtl \sar \cxtr, \rcxtr$}
\RightLabel{$\eqreplii$}
\UnaryInfC{$\rcxtl, a \eq b, \lrf, \cxtl \sar \cxtr, \rcxtr$}
\DisplayProof

& 

\AxiomC{$\rcxtl, a \eq b, a \eq c, b \eq c, \cxtl \sar \cxtr, \rcxtr$}
\RightLabel{$\eqeucl$}
\UnaryInfC{$\rcxtl, a \eq b, a \eq c, \cxtl \sar \cxtr, \rcxtr$}
\DisplayProof
\end{tabular}
\end{center}

\begin{center}
\begin{tabular}{c c c}
\AxiomC{$\rcxtl, \cxtl \sar \cxtr, a \eq b, \rcxtr$}
\RightLabel{$\ineql$}
\UnaryInfC{$\rcxtl, a \ineq b,\cxtl \sar \cxtr, \rcxtr$}
\DisplayProof

&

\AxiomC{$\rcxtl, r(a,b), \cxtl \sar \cxtr, \rcxtr$}
\RightLabel{$\rnegr$}
\UnaryInfC{$\rcxtl, \cxtl \sar \cxtr, \neg r(a,b), \rcxtr$}
\DisplayProof

&

\AxiomC{$\rcxtl, \univ(a,b), \cxtl \sar \cxtr, \rcxtr$}
\RightLabel{$\univl$}
\UnaryInfC{$\rcxtl, \cxtl \sar \cxtr, \rcxtr$}
\DisplayProof
\end{tabular}
\end{center}

\begin{center}
\begin{tabular}{c c c}
\AxiomC{ }
\RightLabel{$\univr$}
\UnaryInfC{$\rcxtl, \cxtl \sar \cxtr, \univ(a,b), \rcxtr$}
\DisplayProof

&

\AxiomC{$\rcxtl, r(a,a), \cxtl \sar \cxtr, \rcxtr$}
\RightLabel{$\selfl$}
\UnaryInfC{$\rcxtl, \cxtl, a : \some \self \sar \cxtr, \rcxtr$}
\DisplayProof

&

\AxiomC{$\rcxtl, a \eq b, \cxtl \sar \cxtr, \rcxtr$}
\RightLabel{$\ineqr$}
\UnaryInfC{$\rcxtl, \cxtl \sar \cxtr, a \ineq b, \rcxtr$}
\DisplayProof
\end{tabular}
\end{center}

We use $\gtalcee$ to denote an extension of a calculus $\gtalce$ with sets of the above rules. We allow for extensions with the sets shown below, and note that the addition of one set of rules may necessitate the addition of another set of rules, as explained above. 
 Extensions with rules for RRAs (such as $\trans{r}$ and $\asy{r}$) are taken into account as extensions with DDRs:
 
\medskip
 
$\{\compl, \compr\}$; $\{\crial, \criar\}$; $\{\selfl, \selfr\}$; $\{\ineql, \ineqr\}$; $\{\rnegl, \rnegr\}$;


$\{\univl, \univr\}$; $\{\invi, \invii,\inviii,\inviv\}$;


$\{\nomli,\nomlii,\nomr,\nomrii\}$; $\{\ltnl, \ltnr, \gtnl, \gtnr\}$; 


$\{\ltunl, \ltunr, \gtunl, \gtunr\}$; $\{\eql, \eqr, \eqeucl, \eqrepli, \eqreplii\}$.

\medskip















\begin{theorem}
\label{thm:G3ALC*-sound-complete}
$\rcxtl, \cxtl \sar \cxtr,\rcxtr$ is derivable in $\gtalcee$ \ifandonlyif \ $\models \rcxtl, \cxtl \sar \cxtr,\rcxtr$.
\end{theorem}

\begin{proof}
Soundness (the forward direction) is shown by induction on the height of the given derivation. Completeness (the backward direction) is shown by a method due to Kripke~\cite{Kri59}. We assume $\rcxtl, \cxtl \sar \cxtr, \rcxtr$ is not derivable, and show that a counter-model can be extracted from failed proof search; thus, if a sequent is not derivable, it is not valid, implying completeness. 
\end{proof}

We additionally show that our calculi possess desirable proof-theoretic properties. Before stating our theorem concerning which properties are possessed, we recall the definition of each property for the reader. 
A rule is defined to be (\emph{height-preserving}) \emph{admissible} in a calculus \ifandonlyif if the premise(s) of the rule is (are) derivable in the calculus (with a certain height), then the conclusion is derivable in the calculus (with a height less than or equal to the height of the premise(s)). Let us define the \emph{inverse} of $(R)$, written $(\hat{R})$, to be the rule obtained by switching the conclusion and the premise(s) of $(R)$. A rule $(R)$ is defined to be (\emph{height-preserving}) \emph{invertible} in a calculus \ifandonlyif $(\hat{R})$ is (height-presevering) admissible. That is, if there exists a derivation for the conclusion, its premises can be derived as well \cite{NegPla11}. As is common in the literature, we usually write \emph{hp-admissible} and \emph{hp-invertible} instead of \emph{height-preserving admissible} and \emph{height-preserving invertible}, and we remark that such properties are important as they can be leveraged to prove decidability of logics~\cite{Kle52}, to permit automated counter-model extraction~\cite{LyoBer19}, or to prove cut-elimination~\cite{NegPla11}, among other applications. Note that in $\sub$, applying a substitution $[b/a]$ to a multiset is defined in the usual way as the replacement of all occurrences of $a$ by $b$ in the multiset. Last, we note that special (hp-)admissible structural rules are shown in \fig~\ref{fig:structural-rules}.

\begin{figure}

\begin{center}
\begin{tabular}{c c c}
\AxiomC{$\rcxtl, \cxtl \sar \cxtr, \rcxtr$}
\RightLabel{$\wkl$}
\UnaryInfC{$\rcxtl, \rcxtl', \cxtl, \cxtl' \sar \cxtr, \rcxtr$}
\DisplayProof

&

\AxiomC{$\rcxtl, \cxtl \sar \cxtr, \rcxtr$}
\RightLabel{$\wkr$}
\UnaryInfC{$\rcxtl, \cxtl \sar \cxtr, \cxtr', \rcxtr, \rcxtr'$}
\DisplayProof

&

\AxiomC{$\rcxtl, \rcxtl', \rcxtl', \cxtl, \cxtl', \cxtl' \sar \cxtr, \rcxtr$}
\RightLabel{$\ctrl$}
\UnaryInfC{$\rcxtl, \rcxtl', \cxtl, \cxtl' \sar \cxtr, \rcxtr$}
\DisplayProof
\end{tabular}
\end{center}

\begin{center}
\begin{tabular}{c c}
\AxiomC{$\rcxtl, \cxtl \sar \cxtr', \cxtr', \cxtr, \rcxtr', \rcxtr', \rcxtr$}
\RightLabel{$\ctrr$}
\UnaryInfC{$\rcxtl, \cxtl \sar \cxtr', \cxtr, \rcxtr', \rcxtr$}
\DisplayProof

&

\AxiomC{$\rcxtl, \cxtl \sar \cxtr, \rcxtr$}
\RightLabel{$\sub$}
\UnaryInfC{$(\rcxtl, \cxtl)[b/a] \sar (\cxtr, \rcxtr)[b/a]$}
\DisplayProof
\end{tabular}
\end{center}

\caption{Admissible structural rules.} 
\label{fig:structural-rules}
\end{figure}

\begin{theorem}\label{thm:admissible-rules} Each calculus $\gtalcee$ possesses the following properties: (i) For all EFs and IFs $X$, $\rcxtl, X, \cxtl \sar \cxtr, X, \rcxtr$ is derivable in $\gtalcee$, (ii) All rules of $\gtalcee$ are hp-invertible, (iii) The $\sub$, $\wkl$, $\wkr$, $\ctrl$, and $\ctrr$ rules are hp-admissible in $\gtalcee$.


\end{theorem}

\begin{proof} 
 (i) is shown by induction on the weight of $X$ (defined in the appendix), and (ii) and (iii) by induction on the height of the given derivation. Details can be found in the appendix.
\end{proof}

\section{Conclusion and Future Work}\label{sec:conclusion}


This paper provides a uniform framework for generating sequent systems on demand for a considerable number of expressive description logics including extensions with role relational axioms. 
All calculi are sound, complete, and possess standard properties. In future work, we aim to optimize our calculi by (i) simplifying the systems through confirming the admissibility of rules (e.g. $\botr$ and $\topl$), (ii) applying a methodology called \emph{structural refinement}~\cite{Lyo21thesis}, which has been used to ready proof systems for use in automated reasoning tasks~\cite{LyoTiuGorClo20,LyoBer19}, and (iii) extending our 
formalism to a broader set of DLs (e.g. intuitionistic or constructive DLs~\cite{FerFioFio10,dePaiva2006ConstructiveDL,uniba21626}) which can be defined proof-theoretically.  
 
 We note that efficient reasoners, based on tableaux, for expressive DLs do already exist (e.g. HermiT~\cite{GliHorMotBorGioWan14}). 
 However, since the current paper merely provides a \emph{framework} for constructing sequent systems for expressive DLs, 
 comparing decision algorithms based on our sequent systems with those based on existing tableaux must be left to future work. 
 Nevertheless, sequent calculi have proven beneficial in establishing meta-logical properties, and thus, we aim to adapt existing methods for sequent systems to obtain constructive proofs of (various forms of) interpolation (as in~\cite{LyoTiuGorClo20,Mae60}), and to utilize our systems in computing re-writings of concepts and TBoxes. Last, we conjecture that cut-elimination holds for $\gtalc$ when we restrict cuts to IFs, though we aim to investigate various forms of cut-elimination for all of our sequent calculi.

\bibliography{bibliography}

\appendix

\section{Proofs}


First, let us provide an example illustrating the closure condition by means of the DDR $\funct{r} \leftrightarrow \forall abc (r(a,b) \land r(a,c) \rightarrow b \eq c)$ defining the functionality RRA $\funct{r}$ for the role $r$.

\begin{example} Suppose we include the following DDR in an extension of $\gtalc$:
\begin{center}
\AxiomC{$\rcxtl, \funct{r}, r(a,b), r(a,c), b \eq c, \cxtl \sar \cxtr, \rcxtr$}
\RightLabel{$\functl$}
\UnaryInfC{$\rcxtl, \funct{r}, r(a,b), r(a,c), \cxtl \sar \cxtr, \rcxtr$}
\DisplayProof
\end{center}
By substituting $a$ for $b$ and $c$, we obtain an instance with two copies of the principal formula $r(a,a)$:
\begin{center}
\AxiomC{$\rcxtl, \funct{r}, r(a,a), r(a,a), a \eq a, \cxtl \sar \cxtr, \rcxtr$}
\RightLabel{$\functl'$}
\UnaryInfC{$\rcxtl, \funct{r}, r(a,a), r(a,a), \cxtl \sar \cxtr, \rcxtr$}
\DisplayProof
\end{center}
By the closure condition, the following rule is required to be in our calculus as well:
\begin{center}
\AxiomC{$\rcxtl, \funct{r}, r(a,a), a \eq a, \cxtl \sar \cxtr, \rcxtr$}
\RightLabel{$\functl''$}
\UnaryInfC{$\rcxtl, \funct{r}, r(a,a), \cxtl \sar \cxtr, \rcxtr$}
\DisplayProof
\end{center}
\end{example}

\medskip

Second, before proving Theorem 1 and Theorem 2, we define the \emph{weight} of a formula.

\begin{definition}[Formula Weight]\label{def:formula-weight} We let $a, b \in \inames$, and define the \emph{weight} of IFs and EFs inductively as shown below:

\begin{itemize}

\item $\fcomp{a \eq b} = 1$

\item $\fcomp{a : \top} = \fcomp{a : \bot} = 1$

\item For $C \in \cnames$, $\fcomp{a : C} = 1$

\item For $r \in \rnames$, $\fcomp{r(a,b)} = 1$

\item $\fcomp{\{a\}} = 2$

\item $\fcomp{a \ineq b} = 2$

\item $\fcomp{\some \self} = 2$

\item $\fcomp{\neg r(a,b)} = 2$

\item $\fcomp{\leqslant n r. \top} = \fcomp{\geqslant n r. \top} = 2$

\item $\fcomp{(r_{1} \comp \cdots \comp r_{n})(a,b)} = \sum_{i = 1}^{n} \fcomp{r_{i}(a,b)} = n$

\item $\fcomp{\rrel{r_{1}, \ldots, r_{l}}} = 1 + \sum_{i = 1}^{n} \fcomp{F_{i}} + \sum_{j = 1}^{k} \fcomp{G_{j}} = 1 + n + k$

\item $\fcomp{a : P \dis Q} = \fcomp{a : P \dis Q} = \fcomp{P \imp Q} = max\{\fcomp{P},\fcomp{Q}\} + 1$

\item $\fcomp{r_{1} \comp \cdots \comp r_{n} \imp r} = \fcomp{(r_{1} \comp \cdots \comp r_{n})(a,b)} + \fcomp{r(a,b)} + 1 = n + 2$

\item $\fcomp{a : \neg P}  = \fcomp{\some P} = \fcomp{\all P} = \fcomp{\leqslant n r. P} = \fcomp{\geqslant n r. P} = \fcomp{P} + 1$

\end{itemize}
\end{definition}

\begin{customthm}{\ref{thm:G3ALC*-sound-complete}}
$\rcxtl, \cxtl \sar \cxtr,\rcxtr$ is derivable in $\gtalcee$ \ifandonlyif $\models \rcxtl, \cxtl \sar \cxtr,\rcxtr$.
\end{customthm}

\begin{proof} To prove completeness we impose a cyclic order $<$ on the rules in $\gtalcee$ and consider each rule in turn. We start with the sequent $\rcxtl, \cxtl \sar \cxtr, \rcxtr$ and apply rules in a bottom-up fashion attempting to construct a proof. Also, we assume a linear order $\prec$ on all individuals $a$. We choose an arbitrary rule $(R)$ in the cyclic order to begin the following procedure:
\begin{itemize}

\item[(1)] If $(R)$ is a rule with its principal formulae occurring in its conclusion, then for each formula in each top-sequent of each open branch of the proof (i.e. a top sequent in the derivation which is not an instance of $\idc$, $\idr$, $\botl$, $\topr$, $\nomrii$, $\eqr$, or $\univr$) for which $(R)$ is bottom-up applicable, apply $(R)$ bottom-up. If the rule makes use of eigenvariables $a_{1}, \ldots, a_{n}$, then select the first $n$ individuals in the linear order $\prec$ that do not yet occur in the derivation to bottom-up apply the rule. 

\item[(2)] Otherwise, $(R)$ is a rule with no principal formula in its conclusion (i.e. $(R)$ is one of the rules $\botr$, $\topr$, $\nomlii$, $\eql$, or $\univl$). For each top-sequent of each open branch of the proof, select the minimal individual(s) $a$ (and $b$) in the linear order $\prec$ such that $(R)$ has not yet introduced $a : \bot$, $a :\top$, $a : \{a\}$, $a \eq a$, or $\univ(a,b)$, respectively, and apply $(R)$ bottom-up.

\item[(3)] If $(R')$ is the next rule in the cyclic order, then set $(R) := (R')$ and go to step (1) above.

\end{itemize}
To simplify our proof, we assume that if no rules are applicable to a top-sequent in an open branch of the derivation being constructed via the above procedure, then we copy the top-sequent an infinite number of times to create an infinite branch. This simplifying assumption is relevant to our application of K\H{o}nig's lemma below.

Let us assume that $\rcxtl, \cxtl \sar \cxtr, \rcxtr$ is not derivable. We will show that $\rcxtl, \cxtl \sar \cxtr, \rcxtr$ is invalid, implying completeness. Since $\rcxtl, \cxtl \sar \cxtr, \rcxtr$ is not derivable, we know that the above procedure will not find a derivation of the sequent, implying (together with our simplifying assumption) that our derivation will be infinitely large. By K\H{o}nig's lemma, we know that an infinitely long path $\branch$ must occur in the derivation since bottom-up applications of rules only permit finite branching. Let $\branch := \{\rcxtl_{i}, \cxtl_{i} \sar \cxtr_{i}, \rcxtr_{i} \ | \ i \in \mathbb{N}\}$ be one such infinite path, and define:
$$
\Theta := \bigcup_{i \in\mathbb{N}} \rcxtl_{i}, \cxtl_{i}
\qquad
\Omega := \bigcup_{i \in\mathbb{N}} \rcxtr_{i}, \cxtr_{i}
$$
We now use $\Theta$ and $\Omega$ to construct an interpretation $\inter = (\dom, \map{\cdot})$ such that $\inter \not\models \rcxtl, \cxtl \sar \cxtr, \rcxtr$. Let $a$ and $b$ be two individuals, and define $a \sim b$ \ifandonlyif $a \eq b \in \Theta$. It is not difficult to show that $a \sim b$ is an equivalence relation: (i) Since the $\eql$ rule will be applied for each individual $a$ in $\branch$, we have that $a \eq a \in \Theta$ for each individual $a$, and (ii) If $a \sim b$ and $a \sim c$, then $a \eq b, a \eq c \in \Theta$, and so by our procedure, at some step $\eqeucl$ will have been applied in the infinite path $\branch$, meaning that $b \eq c \in \Theta$; hence, $b \sim c$ holds. We define $\ec{a} := \{b \in \inames \ | \ a \sim b\}$.

\begin{itemize}

\item $\dom := \{ \ec{a} \ | \ a \text{ occurs in $\branch$.} \}$;

\item For $r \in \rnames$, $(\ec{a}, \ec{b}) \in \map{r}$ \ifandonlyif $r(a,b) \in \Theta$;

\item For $C \in \cnames$, $\ec{a} \in \map{C}$ \ifandonlyif $a : C \in \Theta$;

\item For $a \in \inames$, $\map{a} = \ec{a}$.

\end{itemize}

We extend $\map{\cdot}$ to all complex formulae according to the clauses specified in \sect~\ref{sec:log-prelims}. We now show that (i) if $X \in \Theta$, then $\inter \models X$ and (ii) if $X \in \Omega$, then $\inter \not\models X$. We prove (i) and (ii) by a simultaneous induction on the weight of $X$.\\

\noindent
$a : C$. Let $C \in \cnames$. (i) Suppose that $a : C \in \Theta$. Then, by definition, $\ec{a} \in \map{C}$. (ii) If $a : C \in \Omega$, then since $\branch$ is infinitely long, it cannot be the case that $a : C \in \Theta$ since then $\idc$ would be applied at some point in $\branch$, implying its finiteness. Therefore, $a : C \not\in \Theta$, meaning that $\ec{a} \not\in \map{C}$ by definition.\\

\noindent
$a : \bot$. (i) We know that $a : \bot \not\in \Theta$ since otherwise $a : \bot$ would occur in the antecedent of some sequent along $\branch$, implying that $\botl$ would be applied at some point by our procedure, and $\branch$ would be finite. Hence, the claim follows vacuously. (ii) By the definition of $\inter$, we know that $\map{\bot} = \emptyset$, implying that $\inter \not\models a : \bot$ for all individuals $a$. Hence, the claim follows vacuously.\\

\noindent
$a :\top$. Similar to the previous case.\\

\noindent
$a : \neg P$. (i) Assume that $a : \neg P \in \Theta$. By our procedure, the $\negr$ rule will be applied at some stage $n$ introducing $a : P$ in the consequent of a sequent occurring in $\branch$. By \ih, $\inter \not\models a : P$, implying that $\inter \models \neg P$. (ii) Similar to the proof of claim (i).\\

\noindent
$a : P \dis Q$. (i) Assume that $a : P \dis Q \in \Theta$. Then, by our procedure, the $\disl$ rule will be applied at some point in $\branch$ introducing either $a : P$ or $a : Q$ into the antecedent of some sequent along $\branch$. Let us suppose w.l.o.g. that $a : P$ was introduced. Then, by \ih, we know that $\inter \models a : P$, and so, $\inter \models a : P \dis Q$. (ii) Assume that $a : P \dis Q \in \Omega$. By our procedure, the $\disr$ rule is applied infinitely often, introducing $a : P, a : Q$ into the consequent of $\branch$ at some point. By \ih, it follows that $\inter \not\models a : P$ and $\inter \not\models a : Q$, meaning that $\inter \not\models a : P \dis Q$.\\

\noindent
$a : P \con Q$. Similar to the disjunction case above.\\

\noindent
$P \imp Q$. (i) Let $P \imp Q \in \Theta$. We aim to show that $\inter \models P \imp Q$. Therefore, let us assume that $a^{\inter} \in P^{\inter}$ with the goal of proving that $a^{\inter} \in Q^{\inter}$. By the definition of $\inter$ we know that $a : P \in \Theta$, which implies that $\impl$ will eventually be applied at some point in $\branch$, introducing $a : Q$ into the antecedent of a sequent occurring in $\branch$. Hence, $\inter \models a : Q$ by IH, i.e. $a^{\inter} \in Q^{\inter}$, meaning that $\inter \models P \imp Q$ since $a$ was arbitrary. (ii) Let $P \imp Q \in \Omega$. We aim to show that $\inter \not\models P \imp Q$. Therefore, we aim to show that there exists a $b^{\inter}$ such that $b^{\inter} \in P^{\inter}$, but $b^{\inter} \not\in Q^{\inter}$. By our procedure we know that $\impr$ will eventually be applied at some point in $\branch$, thus introducing $b : P$ into the antecedent of some sequent along $\branch$ and $b : Q$ into the consequent of some sequent along $\branch$, with $b$ fresh. This implies that $b : P \in \Theta$ and $b : Q \in \Omega$, which implies that $\inter \models b : P$ and $\inter \not\models b : Q$ by IH, showing that $\inter \not\models P \imp Q$.\\

\noindent
$a : \some P$. (i) Let $a : \some P \in \Theta$. By our procedure, the $\existsl$ rule will be applied at some point, and will introduce $r(a,b), b : P$ into the antecedent of some sequent along $\branch$ with $b$ fresh. By the $r(a,b)$ case below, we know that $\inter \models r(a,b)$, and by \ih, we know that $\inter \models b : P$. Therefore, $\inter \models \some P$. (ii) Suppose that $a : \some P \in \Omega$. If no relational atom of the form $r(a,b)$ exists in $\Theta$, then the claim follows trivially. Let $r(a,b)$ be an arbitrary relational atom in $\Theta$. By our procedure, the $\existsr$ rule will be applied infinitely often in $\branch$, implying that at some point the formula $b : P$ will be introduced into the consequent of a sequent in $\branch$ for all $r(a,b) \in \Theta$. Hence, if $\inter \models r(a,b)$, then by \ih, $\inter \not\models b : P$. Thus, $\inter \not\models a : \some P$.\\

\noindent
$a : \all P$. Similar to the $\some P$ case above.\\

\noindent
$a : \{b\}$. (i) Suppose that $a : \{b\} \in \Theta$. Then, by the above procedure we know that $\noml$ will be applied infinitely often, and will eventually introduce $a \eq b$ into the antecedent of some sequent in $\branch$, implying that $a \eq b \in \Theta$. Therefore, $\map{a} = \ec{a} = \ec{b} \in \{\ec{b}\} = \{\map{b}\} = \map{\{b\}}$. To complete the proof of the claim, we need to show that $\map{\{b\}}$ is a singleton. We know that $\map{\{b\}}$ has at least one element by definition because $a : \{b\} \in \Theta$. We therefore suppose that $\ec{c}, \ec{d} \in \map{\{b\}}$, and show that $\ec{c} = \ec{d}$. By the definition of $\inter$, we have that $c : \{b\}, d : \{b\} \in \Theta$, meaning that by our procedure $\noml$ and $\eqeucl$ will be applied in $\branch$ and will introduce $c \eq d$. Since $c \eq d \in \Theta$, we have that $\ec{c} = \ec{d}$. Therefore, $\inter \models a : \{b\}$ and $\{b\}$ is a proper nominal. (ii) Suppose that $a : \{b\} \in \Omega$. By the above procedure $\nomr$ will be applied in $\branch$ at some point and will introduce $a \eq b$ into the consequent of some sequent, meaning that $a \eq b \in \Omega$. Since $\branch$ is infinitely long, we know that $a \eq b \not\in \Theta$, since otherwise $\idr$ would be applied and force $\branch$ to be finite. By definition then, $a \not\sim b$, meaning that $\ec{a} \neq \ec{b}$. If we can show that $\map{\{b\}}$ is a singleton, then we know that $\map{a} = \ec{a} \not\in \map{\{b\}}$, and the claim will be proven. Since the $\nomlii$ rule will be applied infinitely often, we know that at some point $b : \{b\}$ will be introduced into the antecedent of a sequent, meaning that $b : \{b\} \in \Theta$, which implies that $\map{b} = \ec{b} \in \{\ec{b}\} = \map{\{b\}}$. Now that we have confirmed that $\map{\{b\}}$ is non-empty, we want to show that all elements it contains are identical, that is, it contains a single element. Let us suppose that $\ec{c}, \ec{d} \in \map{\{b\}}$; we aim to show that $\ec{c} = \ec{d}$. By our supposition $c : \{b\}, d : \{b\} \in \Theta$, and so by our procedure, $\noml$ and $\eqeucl$ will be applied in $\branch$ introducing $c \eq b, d \eq b, c \eq d$ into the antecedent of some sequent in $\branch$, which implies that $\ec{c} = \ec{d}$ since  $c \eq b, c \eq d, c \eq d \in \Theta$. Hence, $\map{\{b\}}$ is a singleton, implying that $\inter \not\models a : \{b\}$.\\

\noindent
$a : (\leqslant n r. P)$. (i) Let $a : (\leqslant n r. P) \in \Theta$. Suppose for $\ec{b_{0}}, \ldots, \ec{b_{n}} \in \dom$ that $\map{r}(\ec{a},\ec{b_{0}}), \ldots, \map{r}(\ec{a},\ec{b_{n}})$ hold. Then, $r(a,b_{0}), \ldots, r(a,b_{n}) \in \Theta$, implying that $\ltnl$ will be applied at some point in $\branch$ introducing either a formula $b_{i} : P$ in the consequent of some sequent of $\branch$, or $b_{i} \eq b_{j}$ in the antecedent of some sequent of $\branch$, for $0 \leq i < j \leq n$. In the former case, $\inter \not\models b_{i} : P$ by \ih, and in the latter case, $\ec{b_{i}} = \ec{b_{j}}$. Hence, for any $n+1$ elements of the domain, either one element does not satisfy $P$ in $\inter$, or two of the elements are identical. It follows that $\inter \models a : (\leqslant n r. P)$. (ii) Let $a : (\leqslant n r. P) \in \Omega$. Then, at some point in $\branch$ we know that $\ltnr$ will be applied and will introduce $r(a,b_{0}), \ldots, r(a,b_{n}), b_{0} : P, \ldots, b_{n} : P$ into the antecedent of some sequent, where all $b_{0}, \ldots, b_{n}$ are fresh. By \ih, we have that $\inter \models r(a,b_{i})$ and $\inter \models b_{i} : P$ for $0 \leq i \leq n$. To prove the claim, we need to additionally show that $\ec{b_{i}} \neq \ec{b_{j}}$ for $0 \leq i < j \leq n$. Since all $b_{i}$ are fresh, it is clear from observing the rules of our calculus that no equality of the form $b_{i} \eq b_{j}$ can be introduced into the antecedent of a sequent in $\branch$. Therefore,  $\ec{b_{i}} \neq \ec{b_{j}}$ for $0 \leq i < j \leq n$, implying that $\inter \not\models a : (\leqslant n r. P)$.\\

\noindent
$a : (\geqslant n r. P)$. Similar to previous case.\\

\noindent
$a : (\leqslant n r. \top)$. Similar to the $a : (\leqslant n r. P)$ case above.\\

\noindent
$a : (\geqslant n r. \top)$. Similar to the $a : (\geqslant n r. P)$ case above.\\

\noindent
$r(a,b)$. (i) Follows from the definition of $\inter$. (ii) Suppose that $r(a,b) \in \Omega$. Then, $r(a,b) \not\in \Theta$, implying that $(\ec{a}, \ec{b}) = (\map{a},\map{b}) \not\in \map{r}$. Hence, $\inter \not\models r(a,b)$.\\

\noindent
$(r \comp s)(a,b)$. (i) Let $(r \comp s)(a,b) \in \Theta$. By our procedure, $\compl$ will be applied at some point in $\branch$, introducing $r(a,c), s(c,b)$ into the antecedent of some sequent of $\branch$ with $c$ fresh. Therefore, there exists a $\ec{c} \in \dom$ such that $\inter \models r(a,c)$ and $\inter \models s(c,b)$, implying that $\inter \models (r \comp s)(a,b)$. (ii) Let $(r \comp s)(a,b) \in \Omega$. By our procedure, $\compr$ will introduce either $r(a,c)$ or $s(c,b)$ for each individual $c$ through continual application in $\branch$, meaning that $r(a,c)$ or $s(c,b)$ will be in $\Omega$ for each individual $c$. Hence, for each $\ec{c} \in \dom$, either $\inter \not\models r(a,c)$ or $\inter \not\models s(c,b)$, meaning that $\inter \not\models (r \comp s)(a,b)$.\\

\noindent
$\inv{r}(a,b)$. (i) Suppose that $\inv{r}(a,b) \in \Theta$. Then, by our procedure, $\invii$ will be applied at some point introducing $r(b,a)$ into the antecedent of some sequent, meaning that $r(b,a) \in \Theta$. Therefore, $\inter \models r(b,a)$, implying that $\inter \models \inv{r}(a,b)$. (ii) Suppose that $\inv{r}(b,a) \in \Omega$. By our procedure, $\inviv$ will be applied in $\branch$ introducing $r(a,b)$ into the consequent of some sequent, implying that $r(a,b) \in \Omega$. Thus, $\inter \not\models r(a,b)$, meaning that $\inter \not\models \inv{r}(b,a)$. We note that the two rules $\invi$ and $\inviii$ are used to conversely establish that $r$ is the inverse of the role $\inv{r}$ occurring in $\branch$.\\

\noindent
$\neg r(a,b)$. Similar to $a : \neg P$ case above.\\

\noindent
$\rrel{r_{1}, \ldots, r_{l}}$. We assume that standard logical connectives such as $\forall$, $\lor$, $\land$, and $\rightarrow$ are interpreted on $\inter$ in the usual way. (i) Suppose that $\rrel{r_{1}, \ldots, r_{l}} \in \Theta$ and let $\inter \models F_{i}$ for each $i \in \{1, \ldots, n\}$. Then, by the definition of $\inter$, we have that $\overline{F} = F_{1}, \ldots, F_{n} \in \Theta$. Hence, at some point in $\branch$ the $\ddrl$ rule will be applied, meaning that for some $j \in \{1, \ldots, k\}$, $G_{j} \in \Theta$. It follows that the definiens of the descriptive definition of $\rrel{r_{1}, \ldots, r_{l}}$ holds in $\inter$, implying that $\inter \models \rrel{r_{1}, \ldots, r_{l}}$. (ii) Suppose that $\rrel{r_{1}, \ldots, r_{l}} \in \Omega$. Then, at some point $\ddrr$ will be applied in $\branch$, introducing $\overline{F}$ into the antecedent of a sequent in $\branch$ and $\overline{G}$ into the consequent of the same sequent with the eigenvariables of the inference fresh. Hence, for some $a_{1}, \ldots, a_{m}$, $\inter \models F_{i}$ for each $i \in \{1, \ldots, n\}$, and $\inter \not\models G_{j}$ for each $j \in \{1, \ldots, k\}$, meaning that the definiens of the descriptive definition of $\rrel{r_{1}, \ldots, r_{l}}$ is not satisfied by $\inter$. This implies that $\inter \not\models \rrel{r_{1}, \ldots, r_{l}}$.\\

\noindent
$r_{1} \comp \cdots \comp r_{n} \imp r$. (i) Let $r_{1} \comp \cdots \comp r_{n} \imp r \in \Theta$. Then, at some point in $\branch$ the $\crial$ rule is applied introducing either $(r_{1} \comp \cdots \comp r_{n})(a,b)$ into the consequent of a sequent in $\branch$ or $r(a,b)$ into the antecedent of a sequent in $\branch$. By \ih, either $\inter \not\models (r_{1} \comp \cdots \comp r_{n})(a,b)$ or $\inter \models r(a,b)$, meaning that $\map{(r_{1} \comp \cdots \comp r_{n})} \subseteq \map{r}$, i.e. $\inter \models r_{1} \comp \cdots \comp r_{n} \imp r$. (ii) Let $r_{1} \comp \cdots \comp r_{n} \imp r \in \Omega$. Then, at some point in $\branch$, the rule $\criar$ will be applied introducing $(r_{1} \comp \cdots \comp r_{n})(a,b)$ into the antecedent of a sequent and $r(a,b)$ into the consequent with both $a$ and $b$ fresh. By \ih, there is some $(a,b)$ such that $\inter \models (r_{1} \comp \cdots \comp r_{n})(a,b)$ and $\inter \not\models r(a,b)$. Hence, $\inter \not\models r_{1} \comp \cdots \comp r_{n} \imp r$.\\

\noindent
$a \eq b$. (i) Suppose that $a \eq b \in \Theta$. Then, by definition we have $a \sim b$, implying that $\map{a} = \ec{a} = \ec{b} = \map{b}$. Therefore, $\inter \models a \eq b$. (ii) Suppose that $a \eq b \in \Omega$. Then, $a \eq b \not\in \Theta$ since otherwise $\idr$ would be applied and $\branch$ would be finite. Hence, $a \not\sim b$, meaning that $\map{a} = \ec{a} \neq \ec{b} = \map{b}$. It follows that $\inter \not\models a \eq b$.\\

\noindent
$a \ineq b$. (i) Suppose that $a \eq b \in \Theta$. Then, at some point in $\branch$ the $\ineql$ rule will be applied, introducing $a \eq b$ into the consequent of a sequent of $\branch$. It follows that $a \eq b \not\in \Theta$ since otherwise $\idr$ would be applied, implying the finiteness of $\branch$. Hence, $a \not\sim b$, meaning that $\map{a} = \ec{a} \neq \ec{b} = \map{b}$, and so, $\inter \models a \ineq b$. (ii) Suppose that $a \eq b \in \Omega$. Then, at some point in $\branch$ the $\ineqr$ rule will be applied, introducing $a \eq b$ into the antecedent of a sequent of $\branch$. Hence, $a \sim b$ holds by definition, meaning that $\map{a} = \ec{a} = \ec{b} = \map{b}$. Consequently, $\inter \not\models a \ineq b$.\\

\noindent
$\univ(a,b)$. (i) By the repeated application of the $\univl$ rule we know that $\univ(a,b)$ will occur in $\Theta$ for each $a$ and $b$. Hence, $\inter \models \univ(a,b)$, meaning that the claim holds. (ii) By the $\univr$ rule we know that $\univ(a,b)$ cannot occur in the consequent of a sequent of $\branch$ since then $\univr$ would be applied and $\branch$ would be finite. Hence, the claim follows vacuously.\\

\noindent
$a : \some \self$. (i) Let $a :\some \self \in \Theta$. Then, at some point in $\branch$, the rule $\selfl$ will be applied, introducing $r(a,a)$ into the antecedent of some sequent in $\branch$. It follows that $\inter \models r(a,a)$, implying that $\inter \models a : \some \self$. (ii) Let $a :\some \self \in \Omega$. Then, at some point in $\branch$, the rule $\selfr$ will be applied, introducing $r(a,a)$ into the consequent of some sequent in $\branch$. It cannot be the case that $r(a,a)$ occurs in $\Theta$ since otherwise it will occur in the antecedent of some sequent of $\branch$, and the $\idr$ rule will be applied at some point ensuring the finiteness of $\branch$. It follows that $\inter \not\models r(a,a)$, implying that $\inter \not\models a : \some \self$.
\end{proof}

\begin{customthm}{\ref{thm:admissible-rules}} Each calculus $\gtalcee$ possesses the following properties:
\begin{itemize}
\item[(i)] For all EFs and IFs $X$, $\rcxtl, X, \cxtl \sar \cxtr, X, \rcxtr$ is derivable in $\gtalcee$.

\item[(ii)] All rules of $\gtalcee$ are hp-invertible.

\item[(iii)] The $\sub$, $\wkl$, $\wkr$, $\ctrl$, and $\ctrr$ rules are hp-admissible in $\gtalcee$. 
\end{itemize}
\end{customthm}

\begin{proof} We argue each claim accordingly: (i) The claim is shown by induction on the weight of $X$. We show the $\{a\}$, $\geqslant n r . P$, $(r_{1} \comp \cdots \comp r_{n})(a,b)$, $\rrel{r_{1}, \ldots, r_{l}}$, and $r_{1} \comp \cdots \comp r_{n} \imp r$ cases; all remaining cases are simple or similar.

\begin{center}
\AxiomC{$\rcxtl, a \eq b, a : \{b\}, \cxtl \sar \cxtr, a : \{b\}, a \eq b, \rcxtr$}
\RightLabel{$\nomr$}
\UnaryInfC{$\rcxtl, a \eq b, a : \{b\}, \cxtl \sar \cxtr, a : \{b\}, \rcxtr$}
\RightLabel{$\nomli$}
\UnaryInfC{$\rcxtl, a : \{b\}, \cxtl \sar \cxtr, a : \{b\}, \rcxtr$}
\DisplayProof
\end{center}

\begin{center}
\begin{center}
\AxiomC{$\Big\{ \rcxtl, \rcxtl', \cxtl, \cxtl' \sar b_{i} : P, a : (\geqslant n r . P), \cxtr, \rcxtr', \rcxtr \ |\ 1 \leq i \leq n \Big\} \cup$}
\noLine
\UnaryInfC{$\Big \{\rcxtl, b_{i} \eq b_{j}, \rcxtl', \cxtl, \cxtl' \sar \cxtr, \rcxtr', \rcxtr \ | \ 0 \leq i < j \leq n \Big \}$}
\RightLabel{$\gtnr$}
\UnaryInfC{$\rcxtl, \rcxtl', \cxtl, \cxtl' \sar a : (\geqslant n r . P), \cxtr, \rcxtr', \rcxtr$}
\RightLabel{$\gtnl$}
\UnaryInfC{$\rcxtl, \cxtl, a : (\geqslant n r . P) \sar a : (\geqslant n r . P), \cxtr, \rcxtr$}
\DisplayProof
\end{center}
\end{center}

$\rcxtl' := r(a,b_{1}), \ldots, r(a,b_{n})$

$\rcxtr' := \{b_{i} \eq b_{j} \ | \ 1 \leq i < j \leq n \}$

$\cxtl' :=  b_{1} : P, \ldots, b_{n} : P$

\begin{center}
\AxiomC{$\Lambda_{1}$}
\AxiomC{$\Lambda_{2}$}
\RightLabel{$\compr$}
\BinaryInfC{$\rcxtl, (r_{1} \comp \cdots \comp r_{n-1})(a,c), r_{n}(c,b), \cxtl \sar \cxtr, (r_{1} \comp \cdots \comp r_{n})(a,b), \rcxtr$}
\RightLabel{$\compl$}
\UnaryInfC{$\rcxtl, (r_{1} \comp \cdots \comp r_{n})(a,b), \cxtl \sar \cxtr, (r_{1} \comp \cdots \comp r_{n})(a,b), \rcxtr$}
\DisplayProof
\end{center}

\begin{flushleft}
$\Lambda_{1} := \rcxtl, (r_{1} \comp \cdots \comp r_{n-1})(a,c), r_{n}(c,b), \cxtl \sar$
\end{flushleft}
\begin{flushright}
$\cxtr, (r_{1} \comp \cdots \comp r_{n})(a,b), (r_{1} \comp \cdots \comp r_{n-1})(a,c), \rcxtr$
\end{flushright}

$\Lambda_{2} := \rcxtl, (r_{1} \comp \cdots \comp r_{n-1})(a,c), r_{n}(c,b), \cxtl \sar \cxtr, (r_{1} \comp \cdots \comp r_{n})(a,b), r_{n}(a,c), \rcxtr$

\begin{center}
\AxiomC{$\Big\{ \rcxtl, \rrel{r_{1}, \ldots, r_{l}}, \overline{F}, G_{j}, \cxtl \sar \cxtr, \overline{G}, \rcxtr \ | \ 1 \leq j \leq k  \Big\}$}
\RightLabel{$\ddrl$}
\UnaryInfC{$\rcxtl, \rrel{r_{1}, \ldots, r_{l}}, \overline{F}, \cxtl \sar \cxtr, \overline{G}, \rcxtr$}
\RightLabel{$\ddrr$}
\UnaryInfC{$\rcxtl, \rrel{r_{1}, \ldots, r_{l}}, \cxtl \sar \cxtr, \rrel{r_{1}, \ldots, r_{l}}, \rcxtr$}
\DisplayProof
\end{center}

\begin{center}
\AxiomC{$\Lambda_{1}$}
\AxiomC{$\Lambda_{2}$}
\RightLabel{$\crial$}
\BinaryInfC{$\rcxtl, r_{1} \comp \cdots \comp r_{n} \imp r, (r_{1} \comp \cdots \comp r_{n})(a,b), \cxtl \sar  \cxtr, r(a,b), \rcxtr$}
\RightLabel{$\criar$}
\UnaryInfC{$\rcxtl, r_{1} \comp \cdots \comp r_{n} \imp r, \cxtl \sar \cxtr, r_{1} \comp \cdots \comp r_{n} \imp r, \rcxtr$}
\DisplayProof
\end{center}

$\Lambda_{1} := \rcxtl, r_{1} \comp \cdots \comp r_{n} \imp r, (r_{1} \comp \cdots \comp r_{n})(a,b), \cxtl \sar  \cxtr, (r_{1} \comp \cdots \comp r_{n})(a,b), r(a,b), \rcxtr$

$\Lambda_{2} := \rcxtl, r_{1} \comp \cdots \comp r_{n} \imp r, (r_{1} \comp \cdots \comp r_{n})(a,b), r(a,b), \cxtl \sar  \cxtr, r(a,b), \rcxtr$\\

\noindent
(ii) The hp-invertibility of $\botr$, $\topl$, $\existsr$, $\alll$, $\compr$, $\nomli$, $\nomlii$, $\nomr$, $\invi$, $\invii$, $\inviii$, $\inviv$, $\ltunl$, $\ltnl$, $\gtunr$, $\gtnr$, $\eql$, $\eqrepli$, $\eqreplii$, $\eqeucl$, $\univl$, and all $\ddrl$ rules follows from the hp-admissibility of $\wkl$ and $\wkr$ shown in (iii) below. All remaining cases are shown by induction on the height of the given derivation and may invoke the hp-admissibility of $\sub$ (in cases where there are eigenvariables) argued in (iii) below.\\

\noindent
(iii) All results are shown by induction on the height of the given derivation. 
 We note that the hp-admissibility of $\wkl$ and $\wkr$ relies on the hp-admissibility of $\sub$, and the hp-admissibility of $\ctrl$ and $\ctrr$ relies on the hp-invertibility of certain rules.

\end{proof}

\end{document}
